\documentclass[12pt]{article}
\usepackage{amssymb,amsthm,amsmath}
\usepackage{graphicx}
\usepackage{fullpage}
\usepackage{subfigure}
\usepackage{graphics}

\newtheorem{theorem}{Theorem}[section]
\newtheorem{lemma}[theorem]{Lemma}

\newtheorem{corollary}[theorem]{Corollary}

\newcommand{\diam}{\mathsf{diam}}

\newcommand{\topp}{\mathsf{top}}
\newcommand{\bottom}{\mathsf{bottom}}
\newcommand{\trees}{\mathsf{Trees}}
\newcommand{\planar}{\mathsf{Planar}}
\newcommand{\len}{\mathsf{len}}
\newcommand{\dil}{\mathsf{dil}}

\renewcommand{\phi}{\varphi}

\title{Optimal stochastic planarization}

\author{Anastasios Sidiropoulos\\
{\small Toyota Technological Institute at Chicago}\\
{\small tasos@ttic.edu}
}
\begin{document}

\setcounter{page}{0}
\maketitle

\begin{abstract}
It has been shown by Indyk and Sidiropoulos \cite{indyk_genus} that any graph of genus $g>0$ can be stochastically embedded into a distribution over planar graphs with distortion $2^{O(g)}$.
This bound was later improved to $O(g^2)$ by Borradaile, Lee and Sidiropoulos \cite{BLS09}.
We give an embedding with distortion $O(\log g)$, which is asymptotically optimal.

Apart from the improved distortion, another advantage of our embedding is that it can be computed in polynomial time.
In contrast, the algorithm of \cite{BLS09} requires solving an NP-hard problem.

Our result implies in particular a reduction for a large class of geometric optimization problems from instances on genus-$g$ graphs, to corresponding ones on planar graphs, with a $O(\log g)$ loss factor in the approximation guarantee.
\end{abstract}

\thispagestyle{empty}
\newpage

\section{Introduction}
\label{sec:intro}

Planar graphs constitute an important class of combinatorial
structures, since they can be used to model a wide variety of natural objects.
At the same time, they have properties that give rise to improved algorithmic solutions for numerous graph problems, if one restricts the set of possible inputs to planar graphs (see, for example \cite{Baker-planar}).

One natural generalization of planarity involves the genus of a graph.  Informally, a graph has genus $g$, for some
$g\geq 0$, if it can be drawn without any crossings on the surface of
a sphere with $g$ additional handles (see Section \ref{sec:prelims}).  For example, a planar graph has genus $0$, and a
graph that can be drawn on a torus has genus at most $1$.

In a way, the genus of a graph quantifies how far it is from
being planar.  Because of their similarities to planar graphs, graphs
of small genus usually exhibit nice algorithmic properties.  More
precisely, algorithms for planar graphs can usually be extended to
graphs of bounded genus, with a small loss in efficiency
or quality of the solution (e.g.~\cite{CEN09}). Unfortunately,
many such extensions are complicated and based on ad-hoc techniques.

Inspired by Bartal's stochastic embedding of general metrics into
trees \cite{Bar96}, Indyk and Sidiropoulos~\cite{indyk_genus} showed that every metric on a graph of genus $g$ can be stochastically embedded into a planar graph with distortion at most exponential in $g$
(see Section \ref{sec:prelims} for a formal definition of
stochastic embeddings).
Since the distortion
measures the ability of the probabilistic mapping to preserve
metric properties of the original space, it is desirable to
make this quantity as small as possible.
The above bound was later improved by Borradaile, Lee, and Sidiropoulos \cite{BLS09}, who obtained an embedding with distortion polynomial in $g$.
In the present paper, we give an embedding with distortion $O(\log g)$, which matches the $\Omega(\log g)$ lower bound from \cite{BLS09}.
The statement of our main result follows.

\begin{theorem}[Stochastic planarization]\label{thm:main}
Any graph $G$ of genus $g$, admits a stochastic embedding into a distribution over planar graphs, with distortion $O(\log g)$.
Moreover, given a drawing of $G$ into a genus-$g$ surface, the embedding can be computed in polynomial time.
\end{theorem}

We note that Theorem \ref{thm:main} can be equivalently stated for compact 2-dimensional simplicial manifolds,
i.e.~continuous spaces obtained by glueing together finitely many triangles, with every point having a neighborhood homeomorphic to a disk.
The shortest-path metrics of genus-$g$ graphs, are precisely the metrics supported on the 0-simplices of such genus-$g$ manifolds.
The result for these spaces can be obtained via a careful affine extension of our embedding over simplices.
Since our focus is on algorithmic applications, we omit the details, and restrict our discussion to the discrete case (i.e.~finite graphs).

\subsection{Our techniques}

In \cite{indyk_genus} it was shown that a graph of genus $g$ can be stochastically embedded into a distribution over graphs of genus $g-1$, with constant distortion.
Repeating this $g$ times results in a planar graph, but yields distortion exponential in $g$.
The improvement of \cite{BLS09} was obtained by giving an algorithm that removes all handles at once.
The main technical tool used to achieve this was the Peeling Lemma from \cite{LS08}.
The idea is that given a graph $G$ of genus $g$, one can find a subgraph $H\subset G$, which we refer to as the \emph{cut graph}, such that (i) $G\setminus H$ is planar, (ii) $H$ has dilation $g^{O(1)}$, and (iii) $H$ can be stochastically embedded into a planar graph.
The resulting distortion of the embedding produced via the Peeling Lemma is proportional to the dilation of $H$, and therefore polynomial in $g$.

It was further shown in \cite{BLS09} that any cut graph has dilation $\Omega(g)$, imposing a limitation on their technique.
We overcome this barrier as follows.
We first find a cut graph consisting of $O(g)$ shortest paths with a common end-point.
These paths are obtained from the generators of the fundamental group of the underlying surface, due to Erickson and Whittlesey \cite{greedy_loops_erickson}.
In the heart of our analysis, we show how to embed a collection of shortest paths with a common end-point, into a random tree with distortion $O(\log g)$.
This result can be viewed as a generalization of the tree-embedding theorem due to Fakcharoenphol, Rao, and Talwar \cite{FRT03}, who showed that any $n$-point metric space admits a stochastic embedding into a tree with distortion $O(\log n)$.

This connection with tree embeddings seems surprising, since planar graphs appear to be significantly more complicated topologically.
For instance, even embedding the $n\times n$ grid into a random tree, requires distortion $\Omega(\log n)$, due to a lower bound of Alon, Karp, Peleg, and West \cite{AKPW}.
Gupta, Newman, Rabinovich, and Sinclair \cite{GNRS} have shown that the same lower bound of $\Omega(\log n)$ holds even for embedding very simple classes of planar graphs into trees, such as series-parallel graphs (i.e.~even for planar graphs of treewidth 2).

Our tree-embedding result is obtained by combining the approach from \cite{FRT03} with the algorithm of Lee and Sidiropoulos \cite{LS2010} for computing random partitions for graphs of small genus.
We remark however that the algorithm of \cite{FRT03} computes an embedding into an \emph{ultrametric}\footnote{A metric space $(X,d)$ where for every $x,y,z\in X$, $d(x,y)\leq \max\{d(x,z),d(z,y)\}$.}, and it can be shown that even a single shortest path cannot be embedded into a random ultrametric with distortion better than $\Omega(\log n)$.
We therefore need new ideas to obtain distortion $O(\log g)$.
One key ingredient towards this is a new random decomposition scheme, which we refer to as \emph{alternating partitions}, and which takes into account the topology of the paths that we wish to partition.
These techniques might be of independent interest.

\subsection{Applications}

\paragraph{Optimization}
As in the case of stochastic embeddings of arbitrary metrics into
trees \cite{Bar96}, we obtain a general
reduction from a class of optimization problems on
genus-$g$ graphs, to their restriction on planar graphs.
We now state precisely the reduction.
Let $V$ be a set,
${\cal I}\subset \mathbb{R}_+^{V\times V}$ a set of non-negative vectors corresponding all feasible solutions for a minimization problem, and $c \in \mathbb{R}_+^{V\times V}$.
Then, we define the \emph{linear minimization problem} $({\cal I}, c)$ to be the computational problem where we are given a graph $G=(V,E)$, and we are asked to find $s\in {\cal I}$, minimizing
\[
\sum_{\{u,v\}\in V\times V} c_{u,v} \cdot s_{u,v} \cdot d(u,v)
\]

Observe that this definition captures a very general class of problems.
For example, MST can be encoded by letting ${\cal I}$ be the set of indicator vectors of the edges of all spanning trees on $V$, and $c$ the all-ones vector.
Similarly, one can easily encode problems such as TSP, Facility-Location, $k$-Server, Bi-Chromatic Matching, etc.

We can now state an immediate Corollary of our embedding result.

\begin{corollary}\label{cor:opt}
Let $\Pi=({\cal I}, c)$ be a linear minimization problem.  If there exists a polynomial-time $\alpha$-approximation algorithm for $\Pi$ on planar graphs, then there exists a randomized polynomial-time $O(\alpha\cdot \log g)$-approximation algorithm for $\Pi$ on graphs of genus $g$.
\end{corollary}

\paragraph{Metric embeddings}
One of the most intriguing open problems in the theory of metric embeddings is determining the optimal distortion for embedding planar graphs, and more generally graphs that exclude a fixed minor, into $L_1$ (see e.g.~\cite{LLR,GNRS, outerplanar,LS08}).
We remark that by the work of Linial, London, and Rabinovich \cite{LLR}, this distortion equals precisely the maximum multi-commodity max-flow/min-cut gap on these graphs, and is therefore of central importance in divide-and-conquer algorithms that are based on Sparsest-Cut \cite{LR,ARV04}.
Our embedding result immediately implies the following corollary.
The first proof of this statement was given in \cite{LS2010}, where it was derived via a fairly complicated argument.
\begin{corollary}\label{cor:emb}
If all planar graphs embed into $L_1$ with distortion at most $\alpha$, then all graphs of genus $g$ embed into $L_1$ with distortion $O(\alpha\cdot \log g)$.
\end{corollary}

\subsection{Preliminaries}\label{sec:prelims}
Throughout the paper, we consider graphs $G=(V,E)$
with a non-negative length function $\len : E \to \mathbb R$.
For a pair $u,v \in V(G)$, we denote the length of the shortest path between $u$ and $v$ in $G$,
with the lengths of edges given by $\len$, by $d_G(u,v)$.
Unless otherwise stated, we restrict our attention
to finite graphs.

\paragraph{Graphs on surfaces}
Let us recall some notions from topological graph theory (an in-depth
exposition can be found in \cite{MoharT-book}).  A \emph{surface} is a
compact connected 2-dimensional manifold, without boundary.
For a graph $G$ we can
define a one-dimensional simplicial complex $C$ associated with $G$ as
follows: The $0$-cells of $C$ are the vertices of $G$, and for each
edge $\{u,v\}$ of $G$, there is a $1$-cell in $C$ connecting $u$ and
$v$.  A \emph{drawing} of $G$ on a surface $S$ is a continuous injection
$f:C\rightarrow V$.
The \emph{orientable genus} of a graph $G$ is the smallest integer
$g\geq 0$ such that $C$ can be drawn into a sphere with $g$ handles.
Note that a graph of genus $0$ is a planar graph.

\paragraph{Metric embeddings}
A mapping $f : X \to Y$ between two metric spaces $(X,d)$ and $(Y,d')$
is {\em non-contracting} if $d'(f(x),f(y)) \geq d(x,y)$ for all $x,y \in X$.
If $(X,d)$ is any finite metric space, and $\mathcal Y$
is a family of finite metric spaces, we say that {\em $(X,d)$ admits a stochastic $D$-embedding into $\mathcal Y$} if there exists a random metric space $(Y,d') \in \mathcal Y$ and a random
non-contracting mapping $f : X \to Y$ such that for every $x,y \in X$,
\begin{equation}
\label{eq:expansion}
\mathbb E\left[\vphantom{\bigoplus} d'(f(x),f(y))\right] \leq D \cdot d(x,y).
\end{equation}
The infimal $D$ such that \eqref{eq:expansion} holds is the {\em distortion of
the stochastic embedding.}
A detailed exposition of  results on metric embeddings can be found in \cite{I-survey} and \cite{Matousek-book}.

\subsection{Organization}
The rest of the paper is organized as follows.
In Section \ref{sec:homo} we show that in any graph of genus $g$, we can find a collection of $O(g)$ shortest paths with a common end-point, whose removal leaves a planar graph.
In Section \ref{sec:alter} we define alternating partitions for the metric space $M$ induced on these paths.
Using these partitions, we show in Section \ref{sec:tree} how to embed $M$ into a random tree, with distortion $O(\log g)$.
Finally, in Section \ref{sec:peeling} we combine this tree embedding with the Peeling Lemma, to obtain our main result.

\section{Homotopy generators}
\label{sec:homo}

Let $G$ be a genus-$g$ graph embedded into an orientable genus-$g$
surface $S$, and let $r$ be a vertex of $G$.
A \emph{system of loops
  with basepoint $r$} is a collection of $2g$ cycles
$C_1,\ldots,C_{2g}$ containing $r$, such that the complement of
$\bigcup_{i=1}^{2g}C_i$ in $S$ is homeomorphic to a disk.
Examples of systems of loops are depicted in figure \ref{fig:loops} (see also \cite{greedy_loops_erickson} for a detailed exposition).
The set of cycles in a system of loops generate the fundamental group $\pi_1(S, r)$.

\begin{figure}
\begin{center}
\scalebox{0.9}{\includegraphics{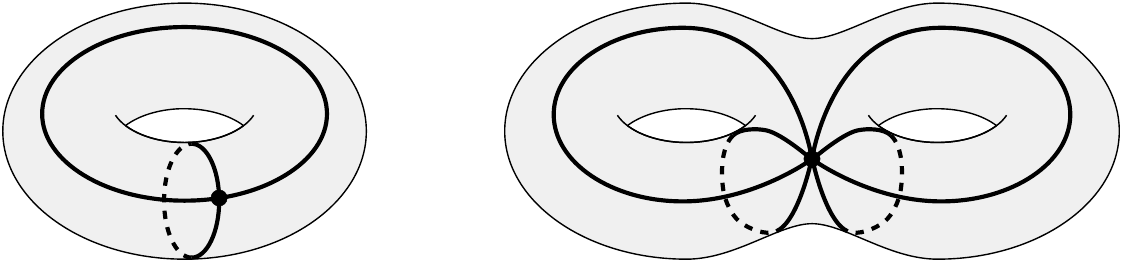}}
\caption{Example of systems of loops for surfaces of genus one and two.\label{fig:loops}}
\end{center}
\end{figure}

A system of loops is called \emph{optimal} if every $C_i$ is
the shortest cycle in its homotopy class.
Algorithms for computing optimal systems of loops have been given by Colin de Verdi\`{e}re and Lazarus \cite{optimal_loops_colin}
and by Erickson and Whittlesey \cite{greedy_loops_erickson}.
The later algorithm has the property that each cycle $C_i$ can be decomposed into either two shortest paths with common end-point $r$, or two such shortest paths, and an edge between the other two end-points.
We therefore have the following.

\begin{lemma}[Greedy homotopy generators \cite{greedy_loops_erickson}]\label{lem:homotopy}
Let $G$ be a graph embedded into an orientable surface $S$ of genus
$g$.  Then, there exists a subgraph $H$ of $G$ satisfying the following properties:
\begin{description}
\item{(i)}
The complement of
$H$ in $S$ is homeomorphic to a disk.
\item{(ii)}
There exists $r\in V(G)$, and a collection of $4g$ shortest-paths $Q_1,\ldots,Q_{4g}$ in $G$,
having $r$ as a common end-point, such that $V(H)=\bigcup_{i\in [4g]} V(Q_i)$.
\end{description}
\end{lemma}

\section{Alternating partitions}
\label{sec:alter}

Let $G$ be a graph.
By rescaling the edge-lengths we may assume w.l.o.g.~that the minimum distance in $G$ is one.
Let ${\cal P}=\{P_1,\ldots,P_k\}$ be a collection of shortest paths in $G$, with a common end-point $r\in V(G)$.
Let
$X=\bigcup_{i=1}^k V(P_i)$.
We consider the metric space $(X,d_G)$.
We define a collection $\{{\cal C}_i\}_{i=0}^{2+\log \Delta}$, where each ${\cal C}_i$ is a random partition of $X$ into sets of diameter less than $2^i$, and such that for any $i\in \{1,\ldots,1+\log \Delta\}$, ${\cal C}_i$ is a refinement of ${\cal C}_{i+1}$.
We refer to the resulting collection $\{{\cal C}_i\}_i$ as \emph{alternating partitions} for $(X, d_G)$.

Pick a permutation $\sigma\in S_g$, and reals\footnote{It suffices to chose $\alpha$ and $\beta$ within $O(\log n)$ bits of precision.} $\alpha\in [0,1)$, $\beta\in [1,2)$, uniformly, and independently at random.
This is all the randomness that will be used in the construction.

We set ${\cal C}_{2+\log \Delta} = \{X\}$, i.e.~the trivial partition that places all points into the same cluster.
For $i=1+\log\Delta,\ldots,1$, given ${\cal C}_{i+1}$ we define ${\cal C}_i$ by performing two partitioning steps that we describe below (see also figure \ref{fig:partition}).
\begin{description}
\item{\textbf{Horizontal partitioning step:}}
Let $A\in {\cal C}_{i+1}$.
We partition $A$ into clusters $\{A_s\}_{s=1}^k$.
We consider the paths in ${\cal P}$ in the order $P_{\sigma(1)},\ldots,P_{\sigma(k)}$.
For each $s\in \{1,\ldots k\}$, we form the cluster
\[
A_{s} = \left(A \cap N_G(P_{\sigma(s)}, \beta\cdot 2^{i-2})\right) \setminus \bigcup_{t=1}^{s-1} A_{t},
\]
where 
$N_G(P_l, \delta) = \{x\in X : d_G(x, P_l) \leq \delta\}$.
We say that the path $P_{\sigma(s)}$ is the \emph{trunk} of $A_{s}$.
For notational convenience, we also refer to $P_{\sigma(1)}$ as the trunk of the unique cluster $\{X\}$ in the partition ${\cal C}_{2+\log\Delta}$.
We refer to the clusters $\{A_{s}\}_{s=1}^k$ as the \emph{horizontal children} of $A$.

\item{\textbf{Vertical partitioning step:}}
Next, we proceed to partition each horizontal child $A_s$ of $A$ into a set of clusters $\{A_{s,j}\}_{j\in \mathbb{N}_0}$, so that for any integer $j\geq 0$,
\[
A_{s,j} = \{x\in A_s : (j-1+\alpha)\cdot 2^{i-2} \leq d_G(r,x) < (j+\alpha)\cdot 2^{i-2}\}
\]
We refer to the clusters $\{A_{s,j}\}_{j\in \mathbb{N}_0}$ as the \emph{vertical children} of $A_s$.
We also say that $P_{\sigma(s)}$ is the trunk of $A_{s,j}$.
Finally, we add all non-empty clusters $A_{s,j}$ to ${\cal C}_i$.
\end{description}

This concludes the description of the construction
of the alternating partitions $\{{\cal C}_i\}_i$ for $(X,d_G)$.

\begin{figure}
\begin{center}
\scalebox{1.0}{\includegraphics{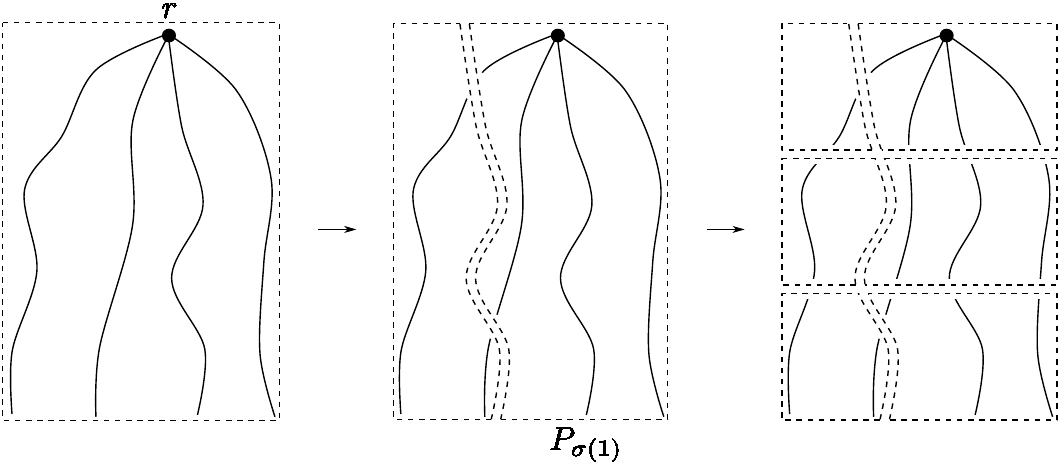}}
\caption{Alternating horizontal and vertical partitioning steps.}\label{fig:partition}
\end{center}
\end{figure}

\begin{lemma}\label{lem:diam}
For any $i\in \{0,\ldots,2+\log\Delta\}$, and $A\in {\cal C}_i$, we have $\diam_G(A) < 2^{i}$.
\end{lemma}
\begin{proof}
Let $P_s$ be the trunk of $A$.
Let $Q$ be the subpath of $P_s$ that is contained in $A$.
By the construction of the vertical children we have $\len(Q) < 2^{i-1}$.
Moreover, by the construction of the horizontal children we have that for any $x\in A$, $d_G(x, Q) < 2^{i-1}$.
Therefore, for any $x,y\in A$ we have $d_G(x,y)\leq d_G(x,C)+d_G(y,C)+\len(Q) < 2^i$.
\end{proof}

\section{Embedding the cut graph into a random tree}
\label{sec:tree}

As in the previous section, let $G$ be a graph, and $r\in V(G)$.
Let $P_1,\ldots,P_k$ be shortest paths in $G$ with common end-point $r$, and define $X=\bigcup_{i=1}^{k} V(P_k)$.
We will use the alternating partitions $\{{\cal C}_i\}_{i}$ constructed in the previous section to obtain a stochastic embedding of $(X, d_G)$ into a distribution over trees, with distortion $O(\log k)$.

For any $Y\subseteq V(G)$, let
$\topp(Y) = \min_{v\in Y} d_G(r, v)$,
and
$\bottom(Y) = \max_{v\in Y} d_G(r, v)$.

We proceed by induction on the partitions $\{{\cal C}_i\}_{i\in \mathbb{N}_0}$, starting from ${\cal C}_0$.
For every cluster $A\in {\cal C}_i$ we construct a tree $T_A$ and an injection $f_A:A\to V(T_A)$.
We inductively maintain the following invariant:
\begin{description}

\item{(I)}
For every cluster $A$ with trunk $P_s$, there exists in $T_A$ a copy of the subpath of $P_s$ containing all vertices $v\in V(P_s)$ with $\topp(A)\leq d_G(r,v)\leq \bottom(A)$.
We refer to this path as the \emph{stem} of $A$.
We denote by $r_A$ the vertex in the stem of $A$ which is closest to $r$ in $G$.
We refer to $r_A$ as the \emph{root} of $A$.
\end{description}

By Lemma \ref{lem:diam} we have that every cluster $A\in {\cal C}_0$ has diameter less than the minimum distance in $G$, and therefore contains a single vertex.
We set $T_A$ to be the trivial tree containing that vertex.
The map $f_A$ sends the unique vertex in $A$ to its copy in $T_A$.

Suppose now that we have constructed a tree for every cluster in ${\cal C}_{i-1}$, for some $i\geq 1$. We will show how to obtain a tree for every cluster in ${\cal C}_i$.
Let $A\in {\cal C}_i$, and 
let $\{A_s\}_{s=1}^k$ be the horizontal children of $A$.
For a horizontal child $A_s$, let $\{A_{s,j}\}_{j\in \mathbb{N}_0}$ be its vertical children.
Recall that each such $A_{s,j}$ is a cluster in ${\cal C}_{i-1}$.
Therefore, by the induction hypothesis we have already computed a tree $T_{A_{s,j}}$ for every $A_{s,j}$, and an injection
$f_{A_{s,j}} : A_{s,j}\to V(T_{A_{s,j}})$.
We construct the tree $T_A$ in two steps:

\begin{description}
\item{\textbf{Vertical composition step:}}
We first combine the graphs of the vertical children of each $A_s$, to obtain an intermediate tree $T_{A_s}$.
This is done as follows.
Recall that $P_s$ is the trunk of $A$.
By the inductive invariant (I) we have that for every vertical child $A_{s,j}$, its stem $Q_{s,j}$ is a path in $T_{A_{s,j}}$, and each such $Q_{s,j}$ is a copy of a subpath of $P_s$.
In particular, since for all $i$ we have $\bottom(A_{s,i}) < \topp(A_{s,i+1})$, it follows that the stems of distinct vertical children correspond to disjoint subpaths of $P_s$.
Let $Q_s$ be the subpath containing all vertices $v\in P_s$ with $\topp(A_s)\leq d_G(r,v)\leq \bottom(A_s)$.
Since $A_{s,j}\subseteq A_s$ it follows that $Q_{s,j}\subseteq Q_s$.
We form the tree $T_{A_s}$ by taking a copy of $Q_s$ and identifying for every vertical child $A_{s,j}$, the stem $Q_{s,j}$ in $T_{A_{s,j}}$ with its copy in $Q_s$.
The path $Q_s$ becomes the stem of $A_s$.
The mapping $f_{A_s}$ is defined by composing each $f_{A_{s,j}}$ with the natural inclusion $V(T_{A_{s,j}}) \to V(T_{A_s})$.

\begin{center}
\scalebox{0.9}{\includegraphics{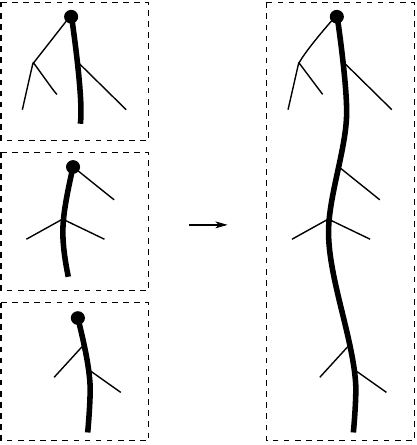}}
\end{center}

\item{\textbf{Horizontal composition step:}}
Next, we combine the trees $T_{A_{s}}$ for all horizontal children $\{A_s\}_{s=1}^k$ of $A$, to obtain $T_A$.
Let $P_t$ be the trunk of $A$.
Observe that there exists a non-empty horizontal child $A_{t}$ of $A$.
For any $l\in \{1,\ldots,k\}$, with $l\neq t$, we connect $r_{A_l}$ with $r_{A_t}$ via an edge of length $2^{i}$.
Let $T_A$ be the resulting tree, and $f_A$ be the induced injection from $\bigcup_{s=1}^k A_{s}$ to $V(T_{A})$.
Note that the root of $A$ is $r_A=r_{A_t}$.

\begin{center}
\scalebox{0.9}{\includegraphics{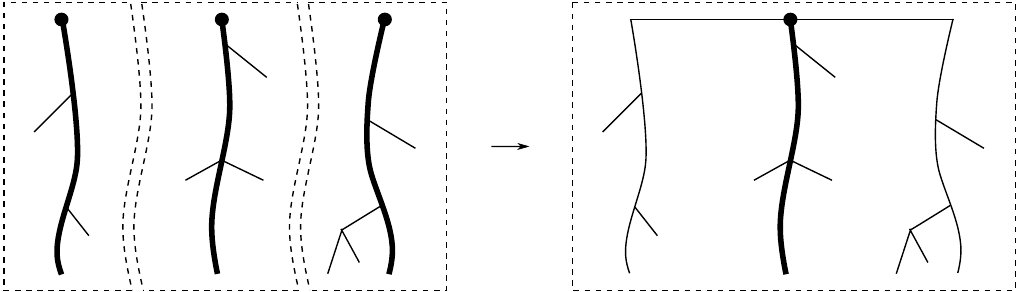}}
\end{center}
\end{description}

It is easy to verity that the inductive invariant (I) is maintained.
Finally, we set $T=T_{\{X\}}$, and $f=f_{\{X\}}$.
This concludes the description of the embedding $f:X\to V(T)$.

\subsection{Bounding the distortion}

It is straight-forward to verify that the mapping $f:X\to V(T)$ is non-contracting, so it remains to bound the expected expansion for every pair of vertices.
We begin with a useful Lemma.

\begin{lemma}\label{lem:stem_H}
Let $i\in \{0,\ldots,2+\log \Delta\}$, let $A\in {\cal C}_i$, and let $Q$ be the stem of $A$.
Then, for any $v\in A$, we have $d_{T}(f(v), Q)\leq 2^{i+2}$.
\end{lemma}
\begin{proof}
For any $j\in \{0,\ldots,i\}$, let $A_j ={\cal C}_j(v)$.
We have that $T_{A_0}$ is the tree containing only $f(v)$.
For any $j\in \{0,\ldots,i\}$, let $Q_j$ be the stem of $A_j$.
We have
$\len(Q_j) < 2^{j-2}$,
and
$d_{T}(r_{A_j}, Q_{j+1}) \leq d_T(r_{A_j},r_{A_{j+1}}) = 2^{j+1}$.
Thus
$d_T(f(v),Q) \leq  \sum_{j=0}^{i-1} \left( \len(Q_j) + d_T(r_{A_j},Q_{j+1}) \right) < \sum_{j=0}^{i-1} \left( 2^{j-2} + 2^{j+1} \right) <  2^{i+2}$
\end{proof}

For the remaining of the analysis, we fix two vertices $u,v\in X$.
We wish to bound $\mathbb{E}[d_T(f(u), f(v))]$, where the expectation is taken over the randomness used in constructing the alternating partitions $\{{\cal C}_i\}_i$ (i.e.~$\alpha$, $\beta$, and $\sigma$).

We begin by introducing some notation.
We say that a path $P_s\in {\cal P}$ \emph{settles} $\{u,v\}$ at level $i$ if $u$ and $v$ are in the same cluster in ${\cal C}_{i+1}$, and $P_s$ is the first path w.r.to the ordering $\sigma$ such that $P_s$ is the trunk of at least one of the clusters ${\cal C}_i(u)$, ${\cal C}_i(v)$.
Moreover, we say that $P_s$ \emph{cuts horizontally} $\{u,v\}$ at level $i$ if it settles $\{u,v\}$ at level $i$, and exactly one of the clusters ${\cal C}_i(u)$, ${\cal C}_i(v)$ has $P_s$ as its trunk.

Similarly, we say that $P_s$ \emph{saves} $\{u,v\}$ at level $i$ if $u$ and $v$ are in the same cluster in ${\cal C}_{i+1}$, and 
$P_s$ is the trunk of a cluster in ${\cal C}_i$ containing both $u$ and $v$.
We say that $P_s$ \emph{cuts vertically} $\{u,v\}$ at level $i$ if it saves $\{u,v\}$ at level $i$, there is a horizontal child of a cluster in ${\cal C}_i$ containing both $u$ and $v$, and ${\cal C}_{i+1}(u)\neq {\cal C}_{i+1}(v)$.

Let $\gamma_s^i$, resp.~$\delta_s^i$, be the supremum of $d_T(f(u),f(v))$ when $s$ cuts $\{u,v\}$ at level $i$ horizontally, resp.~vertically, taken over all possible random choices of the algorithm.
That is,
\[
\gamma_s^i = \sup_{\alpha,\beta,\sigma} \left\{d_T(f(u),f(v)) : s \mbox{ cuts horizontally } \{u,v\} \mbox{ at level } i\right\}
\]
\[
\delta_s^i = \sup_{\alpha,\beta,\sigma} \left\{d_T(f(u),f(v)) : s \mbox{ cuts vertically } \{u,v\} \mbox{ at level } i\right\}
\]

Then, we have
\begin{equation}
\mathbb{E}[d_T(f(u),f(v))]  \leq  \Phi_1 + \Phi_2, \label{eq:phi1phi2}
\end{equation}
where
\[
\Phi_1 = \sum_{s=1}^k \sum_{i=0}^{2+\log\Delta} \gamma_s^i \cdot \Pr[P_s \mbox{ cuts horizontally } \{u,v\} \mbox{ at level } i]
\]
\[
\Phi_2 = \sum_{s=1}^k \sum_{i=0}^{2+\log\Delta} \delta_s^i \cdot \Pr[P_s \mbox{ cuts vertically } \{u,v\} \mbox{ at level } i]
\]
We will bound each one of these quantities separately.

\begin{lemma}\label{lem:phi1}
$\Phi_1 \leq O(\log k) \cdot d_G(u,v)$.
\end{lemma}
\begin{proof}
Define the interval
\[
I_s = [\min\{d_G(u,P_s),d_G(v,P_s)\}, \max\{d_G(u,P_s),d_G(v,P_s)\})
\]
In order for $P_s$ to cut horizontally $\{u,v\}$ at level $i$ it must be the case that
$\beta \cdot 2^{i-2} \in I_s$.
Since $\beta$ is chosen from $[1,2)$ uniformly at random, it follows by the triangle inequality that this happens with probability at most
\begin{equation}
\Pr[\beta \cdot 2^{i-2} \in I_s]\leq |I_s|/2^{i-2} \leq d_G(u,v)/2^{i-2} \label{eq:ph1_1}
\end{equation}

Assume w.l.o.g.~that $d_G(P_1,\{u,v\}) \leq \ldots \leq d_G(P_k, \{u,v\})$.
Conditioned on the event that $\beta\cdot 2^{i-2}\in I_s$, any of the paths $P_1,\ldots,P_s$ can settle $\{u,v\}$.
Therefore,
\begin{equation}
\Pr[P_s \mbox{ settles } \{u,v\} \mbox{ $|$ } \beta\cdot 2^{i-2}\in I_s] \leq 1/s \label{eq:phi1_2}
\end{equation}

Next we bound $\gamma_s^i$.
Suppose that a path $P_s$ cuts horizontally $\{u,v\}$ at level $i$.
Let $Q$ be the stem of the cluster in ${\cal C}_{i+1}$ containing both $u$ and $v$.
By Lemma \ref{lem:stem_H} we conclude that 
\begin{equation}
\gamma_s^i \leq d_T(f(u), f(v)) \leq d_T(f(u),Q)+\len(Q)+d_T(f(v),Q) \leq 2^{i+5} \label{eq:phi1_3}
\end{equation}

Observe that since $\beta\in [1,2)$, it follows that for every $s\in \{1,\ldots,k\}$, the path $P_s$ can cut $\{u,v\}$ only at a single level $i_s$.
Therefore
\begin{eqnarray*}
\Phi_1 & \leq & \sum_{s=1}^k 2^{i_s+5} \cdot \Pr[P_s \mbox{ settles } \{u,v\} \mbox{ at level $i_s$ and } \beta \cdot 2^{i_s-2} \in I_s]\\
 & \leq & \sum_{i=1}^k 2^{i_s+5} \cdot \Pr[P_s \mbox{ settles } \{u,v\} \mbox{ at level $i_s$ $|$ } \beta \cdot 2^{i_s-2} \in I_s] \cdot \Pr[\beta \cdot 2^{i_s-2} \in I_s]\\
 & \leq & \sum_{s=1}^k 2^{i_s+5} \cdot \frac{1}{s} \cdot \frac{d_G(u,v)}{2^{i_s-2}}\\
 & = & O(\log k) \cdot d_G(u,v)
\end{eqnarray*}
\end{proof}

\begin{lemma}\label{lem:phi2}
$\Phi_2 \leq O(\log k) \cdot d_G(u,v)$.
\end{lemma}
\begin{proof}
Define
\[
J_s = [\max\{d_G(u,P_s), d_G(v,P_s)\}, \infty)
\]
\[
R_s^i = \bigcup_{j=0}^\infty \left[j \cdot 2^{i-2} + \min\{d_G(r,u),d_G(r,v)\}, j \cdot 2^{i-2} + \max\{d_G(r,u),d_G(r,v)\}\right)
\]
Denote by ${\cal E}_1^i$ the event $\beta\cdot 2^{i-2}\in J_s$, and by ${\cal E}_2^i$ the event $\alpha\cdot 2^{i-2}\in R_s^i$.
In order for $P_s$ to cut vertically $\{u,v\}$ at level $i$, both ${\cal E}_1^i$ and ${\cal E}_2^i$ must hold.

Assume w.l.o.g.~that $d_G(P_1,\{u,v\}) \leq \ldots \leq d_G(P_k, \{u,v\})$.
Conditioned on ${\cal E}_1^i$, any of $P_1,\ldots,P_s$ can save $\{u,v\}$.
Therefore, 
\begin{equation}
\Pr[P_s \mbox{ saves } \{u,v\} \mbox{ $|$ } {\cal E}_1^i] \leq 1/s \label{eq:phi2_1}
\end{equation}

By the triangle inequality we have
\begin{equation}
\Pr[{\cal E}_2^i] \leq d_G(u,v)/2^{i-2} \label{eq:phi2_2}
\end{equation}

We next upper bound $\delta_s^i$.
Suppose that $P_s$ cuts vertically $\{u,v\}$ at level $i$.
Let
$L_s = \max\{d_G(u,P_s), d_G(v,P_s))\}$,
and 
$j_s = 2 + \left\lceil\log L_s\right\rceil$.
Assume w.l.o.g.~that $d_G(r,u)\leq d_G(r,v)$.
If follows by the construction of $T$ that $P_s$ is the trunk of ${\cal C}_{j_s+1}(u)$ and ${\cal C}_{j_s+1}(v)$.
Therefore, there exist clusters $U_1,\ldots,U_t \in {\cal C}_{j_s+1}$ with $u\in U_1$, $v\in U_t$, such that $P_s$ is the trunk of every $U_i$, and 
such that the stems of $U_1,\ldots,U_t$ are consecutive subpaths of $P_s$.
Let $Q_i$ be the stem of $U_i$.
For any $j\in \{1,\ldots,t-1\}$, the path $Q_j$ is connected to path $Q_{j+1}$ via a path $W_j$, such that $\len(Q_j) + \len(W_j) \leq 2^{j_s-1}$.
By lemma \ref{lem:stem_H} we have
\begin{eqnarray}
\delta_s^i & \leq & d_T(f(u),Q_1) + d_T(f(v),Q_t) + \sum_{j=1}^{t-1} (\len(Q_j)+\len(W_j)) + \len(Q_t) \nonumber \\
 & \leq & 2^{j_s+4} + |d_G(r,u)-d_G(r,v)| + 2\cdot 2^{j_s-1} \nonumber \\
 & \leq & 2^8 \cdot L_s + d_G(u,v) \label{eq:phi2_3}
\end{eqnarray}

Let
$\tau_s = \max\{j_s, \lceil \log d_G(u,v)\rceil - 1\}$.
Observe that for any $i<j_s$, we have $\Pr[{\cal E}_1^i]=0$.
Moreover, if $P_s$ saves $\{u,v\}$ at level $i$, it must be the case that there exists $A\in {\cal C}_{i+1}$ containing both $u$ and $v$, and thus with $\diam_G(A)\geq d_G(u,v)$.
By Lemma \ref{lem:diam} we have $i \geq \lceil\log d_G(u,v)\rceil - 1$.
Therefore, for any $i<\tau_s$ we have $\Pr[P_s \mbox{ cuts } \{u,v\} \mbox{ at level } i]=0$.
Combining with (\ref{eq:phi2_1}), (\ref{eq:phi2_2}), and (\ref{eq:phi2_3}), we have 
\begin{eqnarray*}
\Phi_2 & \leq & \sum_{s=1}^k \sum_{i=\tau_s}^{2 + \log \Delta} \delta_s^i \cdot \Pr[P_s \mbox{ saves } \{u,v\} \mbox{ at level } i \mbox{ $|$ } {\cal E}_1^i \mbox{ and } {\cal E}_2^i] \cdot \Pr[{\cal E}_1^i \mbox{ and } {\cal E}_2^i]\\
 & \leq & \sum_{s=1}^k \frac{1}{s} \sum_{i=\tau_s}^{2 + \log \Delta} \delta_s^i \cdot \frac{d_G(u,v)}{2^{i-2}}\\
 & \leq & d_G(u,v) \cdot \sum_{s=1}^k \frac{1}{s} \sum_{i=\tau_s}^{2 + \log \Delta} \frac{2^8\cdot L_s + d_G(u,v)}{2^{i-2}}\\
 & < & d_G(u,v) \cdot \sum_{s=1}^k \frac{1}{s} \cdot  \frac{2^8\cdot L_s + d_G(u,v)}{2^{\tau_s-3}}\\
 & \leq & O(\log k) \cdot d_G(u,v)
\end{eqnarray*}
\end{proof}

Combining (\ref{eq:phi1phi2}) with lemmas \ref{lem:phi1} and \ref{lem:phi2} we obtain the main result of this section.

\begin{theorem}\label{thm:tree_embed}
Let $G$ be a graph, and let ${\cal P}=\{P_1,\ldots,P_k\}$ be a collection of shortest-paths in $G$, sharing a common end-point.
Then, the metric space $\left(\bigcup_{i=1}^{k} V(P_i), d_G\right)$ admits a stochastic embedding into a distribution over trees with distortion $O(\log k)$.
\end{theorem}

\section{Planarization}
\label{sec:peeling}

Let $(X,d)$ be a metric space.
A distribution ${\cal F}$ over partitions of $X$ is called \emph{$(\beta,\Delta)$-Lipschitz} if every partition in the support of ${\cal F}$ has only clusters of diameter at most $\Delta$, and for every $x,y\in X$, 
\[
\Pr_{C\in {\cal F}}[C(x)\neq C(y)] \leq \beta \cdot \frac{d(x,y)}{\Delta}.
\]

We denote by $\beta_{(X,d)}$ the infimum $\beta$ such that for any $\Delta>0$, the metric $(X,d)$ admits a $(\Delta, \beta)$-Lipschitz random partition, and we refer to $\beta_{(X,d)}$ as the \emph{modulus of decomposability} of $(X,d)$.
The following theorem is due to Klein, Plotkin, and Rao \cite{KPR93}, and Rao \cite{Rao99}.
\begin{theorem}[\cite{KPR93}, \cite{Rao99}]\label{thm:KPR}
For any planar graph $G$, we have $\beta_{(V(G), d_G)} = O(1)$.
\end{theorem}

Let $G$ be a graph, and let $A\subseteq V(G)$.
The \emph{dilation} of $A$ is defined to be
\[
\dil_G(A) = \max_{u,v\in V(G)} \frac{d_{G[A]}(u,v)}{d_G(u,v)}
\]

For a graph $G$ and a graph family ${\cal F}$ we write $G\overset{D}{\leadsto} {\cal F}$ to denote the fact that $G$ stochastically embeds into a distribution over graphs in ${\cal F}$, with distortion $D$.

For two graphs $G,G'$, a \emph{1-sum} of $G$ with $G'$ is a graph obtained by taking two disjoint copies of $G$ and $G'$, and identifying a vertex $v\in V(G)$ with a vertex $v'\in V(G')$.
For a graph family ${\cal X}$, we denote by $\oplus_1{\cal X}$ the closure of ${\cal X}$ under 1-sums.

\begin{lemma}[Peeling Lemma \cite{LS08}]\label{lem:peeling}
Let $G$ be a graph, and $A\subseteq V(G)$.
Let $G'=(V(G),E')$ be a graph with $E'=E(G)\setminus E(G[A])$,
and let $\beta = \beta_{(V, d_{G'})}$ be the corresponding modulus of decomposability.
Then, there exists a graph family ${\cal F}$ such that $G\overset{D}{\leadsto} {\cal F}$, where $D=O(\beta \cdot \dil_G(A))$, and every graph in ${\cal F}$ is a 1-sum of isometric copies of the graphs $G[A]$ and $\left\{G[V\setminus A\cup \{a\}]\right\}_{a\in A}$.
\end{lemma}

We will use the following auxiliary Lemma.

\begin{lemma}[Composition Lemma]\label{lem:composition}
Let $G$ be a graph, and let ${\cal X}$, ${\cal Y}$, ${\cal Z}$ be graph families.
If $G\overset{\alpha}{\leadsto} \oplus_1 ({\cal X}\cup {\cal Y})$, and for any $G'\in {\cal X}$, $G'\overset{\beta}{\leadsto} {\cal Z}$, then
$G\overset{\alpha\cdot \beta}{\leadsto} \oplus_1 ({\cal Z}\cup {\cal Y})$.
\end{lemma}
\begin{proof}[Proof sketch]
It follows by direct composition of the two embeddings.
\end{proof}

\begin{proof}[Proof of Theorem \ref{thm:main}]
Let $G$ be a genus-$g$ graph, drawn on a genus-$g$ surface $S$.
Let $H$ be the subgraph of $G$ given by Lemma \ref{lem:homotopy}.
Recall that there exists $r\in V(G)$ and shortest paths $P_1,\ldots,P_{4g}$ in $G$, having $r$ as a common end-point, such that $V(H) = \bigcup_{i=1}^{4g} V(P_i)$.

Let us write $G_1 = G$, and $X = V(H)$.
By Theorem \ref{thm:tree_embed}, we have
\begin{equation}
(X,d_{G_1}) \overset{O(\log g)}{\leadsto} \trees \label{eq:first_tree}
\end{equation}

After scaling the lengths of the edges in $G_1$, we may assume that the minimum distance is one.
Let $G_2$ be the graph obtained from $G_1$ as follows.
For every edge $\{u,v\}\in E(G_1)$ with $u\in X_1$ and $v\notin X_1$, we replace $\{u,v\}$ with two edges $\{u,w\}$ and $\{w,v\}$, with $\len(\{u,w\})=1/2$, and $\len(\{w,v\})=\len(\{u,v\})-1/2$.
Let $Y=X\cup N_{G_2}(X)$, i.e.~the set $X$, together with all new vertices $w$ introduced above.

Observe that for any $x,y\in X$, we have
$d_{G_1}(x,y) = d_{G_2}(x,y)$.
Therefore by (\ref{eq:first_tree}),
\[
(X, d_{G_2}) \overset{O(\log g)}{\leadsto} \trees
\]
This embedding can be extended to $Y$ as follows.
For every tree $T$ in the support of the distribution, and for every vertex $w\in Y\setminus X$, we attach $w$ to $T$ by adding an edge of length $1/2$ between $w$ and the unique neighbor of $w$ in $X$.
Since we only add leaves to $T$, the new graph is still a tree.
It is straight-forward to verify that the resulting stochastic embedding has distortion $O(\log g)$, and thus
\begin{equation}
(Y, d_{G_2}) \overset{O(\log g)}{\leadsto} \trees \label{eq:G3}
\end{equation}

Let $G_3$ be the graph obtained from $G_2$ by adding an edge $\{u,v\}$ of length $d_G(u,v)$, between every pair of vertices $u,v\in Y$.
Let $E_3' = E(G_3) \setminus E(G_3[Y])$, and $G_3'=(V(G_3), E_3')$.
By cutting the surface $S$ along $H$ we obtain a drawing of $G\setminus H$ into the interior of a disk.
Since every vertex in $Y\setminus X$ is attached to $G\setminus H$ via a single edge, it follows that this planar drawing can be extended to $G_3'$.
Thus, the graph $G_3'$ is planar, and by theorem \ref{thm:KPR} we have $\beta_{(V(G_3'), d_{G_3'})} = O(1)$.

Similarly, we have that for any $y\in Y$, the graph $G_3[V(G_3) \setminus Y \cup \{y\}]$ is planar.

Observe that $\dil_{G_3}(Y)=1$.
By the Peeling Lemma (Lemma \ref{lem:peeling}) we have that $G_3$ can be stochastically embedded with distortion $O(\dil_{G_3}(Y)\cdot \beta_{G_3'}) = O(1)$ into a distribution over graphs $J$, where $J$ is obtained by 1-sums of isometric copies of $G_3[Y_3]$ and $\{G_3[V(G_3) \setminus Y \cup \{y\}]\}_{y\in Y}$.
Each graph in $\{G_3[V(G_3) \setminus Y \cup \{y\}]\}_{y\in Y}$ is planar, and therefore
\begin{equation}
G_3 \overset{O(1)}{\leadsto} \oplus_1(\planar, \{G_3[Y_3]\}) \label{eq:planar_clique}
\end{equation}
By (\ref{eq:G3}), and since the metric $(Y_3, d_{G_3})$ is the same as the metric $(Y_3, d_{G_2})$, we have that
\begin{equation}
G_3[Y_3] \overset{O(\log g)}{\leadsto} \trees \label{eq:Y_trees}
\end{equation}
Note that the 1-sum of two planar graphs is also planar.
Therefore, combining (\ref{eq:planar_clique}), (\ref{eq:Y_trees}) and Lemma \ref{lem:composition}, we obtain
\[
G_3 \overset{O(\log g)}{\leadsto} \oplus_{1}(\trees \cup \planar) = \oplus_1\planar = \planar
\]
Since $G_3$ contains an isometric copy of $G$, this implies
\[
G\overset{O(\log g)}{\leadsto} \planar,
\]
concluding the proof.
\end{proof}

\paragraph{Acknowledgements}
The author wishes to thank Glencora Borradaile, Piotr Indyk, and James R.~Lee for their invaluable contributions to the development of the techniques used in this paper.
He also thanks Dimitrios Thilikos for a question which motivated the problem studied here, as well as Jeff Erickson, and Sariel Har-Peled for several enlightening discussions.

\bibliography{bibfile}
\bibliographystyle{alpha}

\end{document}